\documentclass[12pt,reqno,a4paper]{amsart}
\usepackage{amssymb,amscd}

\usepackage{hyperref}
\begin{document}

\let\kappa=\varkappa
\let\eps=\varepsilon
\let\phi=\varphi
\let\p\partial

\def\Z{\mathbb Z}
\def\R{\mathbb R}
\def\C{\mathbb C}
\def\Q{\mathbb Q}
\def\P{\mathbb P}
\def\HH{\mathrm{H}}
\def\ss{{X}}

\def\conj{\overline}
\def\Beta{\mathrm{B}}
\def\const{\mathrm{const}}
\def\ov{\overline}

\renewcommand{\Im}{\mathop{\mathrm{Im}}\nolimits}
\renewcommand{\Re}{\mathop{\mathrm{Re}}\nolimits}
\newcommand{\codim}{\mathop{\mathrm{codim}}\nolimits}
\newcommand{\id}{\mathop{\mathrm{id}}\nolimits}
\newcommand{\Aut}{\mathop{\mathrm{Aut}}\nolimits}
\newcommand{\alk}{\mathop{\mathrm{alk}}\nolimits}
\newcommand{\lk}{\mathop{\mathrm{lk}}\nolimits}
\newcommand{\Ker}{\mathop{\mathrm{Ker}}\nolimits}
\newcommand{\sign}{\mathop{\mathrm{sign}}\nolimits}
\newcommand{\rk}{\mathop{\mathrm{rk}}\nolimits}
\def\Jet{{\mathcal J}}

\renewcommand{\mod}{\mathrel{\mathrm{mod}}}

%%%%  Tchernoff's defs - on

%%%%  Tchernoff's defs - off

\newtheorem{mainthm}{Theorem}
\renewcommand{\themainthm}{{\Alph{mainthm}}}
\newtheorem{thm}{Theorem}[subsection]
\newtheorem{lem}[thm]{Lemma}
\newtheorem{prop}[thm]{Proposition}
\newtheorem{cor}[thm]{Corollary}

\theoremstyle{definition}
\newtheorem{exm}[thm]{Example}
\newtheorem{rem}[thm]{Remark}
\newtheorem{df}[thm]{Definition}

\renewcommand{\thesubsection}{\arabic{subsection}}

\title[Affine linking number estimates]{Affine linking number estimates for the number of times an observer sees a star}
\author{Vladimir Chernov and Ryan Maguire}
\address{6188 Kemeny Hall, Department of Mathematics, Dartmouth College, Hanover, NH 03755, USA}
\email{Vladimir.Chernov@dartmouth.edu}
\address{6188 Kemeny Hall, Department of Mathematics, Dartmouth College, Hanover, NH 03755, USA}
\email{Ryan.J.Maguire.GR@dartmouth.edu}
\begin{abstract}
Affine linking numbers are the generalization of linking numbers to the case of nonzero homologous linked submanifolds. They were introduced by Rudyak and the first author who used them to study causality in globally hyperbolic spacetimes. 

In this paper we use affine linking numbers to estimate the number of times an observer sees light from a star, that is how many copies of the star do they see on the sky due to gravitational lensing.
\end{abstract}


\maketitle

We work in the $C^{\infty}$-category and the word {\it smooth\/} means $C^{\infty}.$ All the manifolds, maps
etc.~are assumed to be smooth unless the opposite is explicitly stated. The manifolds are assumed to be oriented.

\subsection{Lorentz geometry: conventions and definitions\/}
In this section we introduce the basic conepts of Lorentz geometry following the paper of  Nemirovski and the first author~\cite{ChernovNemirovski}.

Let $(\ss^{m+1}, g)$ be an $(m+1)$-dimensional Lorentz manifold and $p\in \ss$. A nonzero ${\bf v}\in T_p\ss$ 
is called {\it spacelike, timelike, nonspacelike, or null\/} if $g({\bf v}, {\bf v})$ is respectively 
positive, negative, non-positive or   zero. A piecewise smooth curve is timelike if all of its velocity 
vectors are timelike. Nonspacelike, spacelike and null curves are defined similarly. Since $(\ss, g)$ has a unique 
Levi-Cevita connection, see for example~\cite[page 22]{BeemEhrlichEasley}, we can talk about spacelike, 
timelike and null geodesics. A submanifold $M\subset \ss$ is {\it spacelike} if $g$ restricted to $TM$ is a 
Riemann metric.

The set of all nonspacelike vectors in $T_p\ss$ is a cone consisting of two hemicones, and the continuous 
with respect to $p\in \ss$ choice of one of the two hemicones is called the {\it time orientation\/} of
$(\ss, g)$. The vectors from the chosen hemicones are called {\it future pointing.\/}
A time oriented Lorentz manifold is called a {\it spacetime\/} and the points of it are called events

For $x$ in a spacetime $(\ss, g)$ its {\it causal future\/} $J^+(x)\subset \ss$ 
 is the set of all $y\in \ss$ that can be reached by a 
future pointing nonspacelike curve from $x.$ 
The causal past $J^-(x)$ of the event $x\in \ss$ is defined similarly. 

Two events  $x,y$ are said to be {\it causally related\/} if $x\in J^+(y)$ or $y\in J^+(x).$

A spacetime is said to be {\it globally hyperbolic\/} if $J^+(x)\cap J^-(y)$ is compact for every $x,y\in \ss$
and if it is {\it causal,\/} i.e.~it has no closed nonspacelike curves. The classical definition 
of a global hyperbolicity requires $(\ss,g)$ to be strongly causal rather than just causal, see~\cite{HawkingEllis}, however
Bernal and Sanchez~\cite[Theorem 3.2]{BernalSanchezCausal} proved that that two definitions are equivalent.  The compactness condition is known as {\em absence of naked singularities.}

A {\it Cauchy surface\/} in $(\ss, g)$ is a subset such that every inextendible nonspacelike curve 
$\gamma(t)$ intersects it at exactly one value of $t.$ A classical result, see~\cite[pages 211-212]{HawkingEllis}, is that $(\ss, g)$ is globally 
hyperbolic if and only if it has a Cauchy surface. 
Geroch~\cite{Geroch} proved that every globally hyperbolic $(\ss, g)$ is homemorphic to a product of a Cauchy surface $M$ and $\R$ with each $M\times t$ being a Cauchy surface. Bernal and Sanchez~\cite[Theorem 1]{BernalSanchez},~\cite[Theorem
1.1]{BernalSanchezMetricSplitting},~\cite[Theorem 1.2]{BernalSanchezFurther} proved much more namely that every globally 
hyperbolic $(\ss^{m+1}, g)$ has a smooth spacelike Cauchy surface $M^m$ and that moreover
for every smooth spacelike Cauchy surface $M$ there is a
diffeomorphism  $h:M\times \R\to \ss$ such that 
\begin{description}
\item[a] $h(M\times t)$ is a smooth spacelike Cauchy surface
for all $t$, 
\item[b] $h(x\times \R)$ is a future pointing timelike curve for all $x\in M$, and finally 
\item[c] $h(M\times 0)=M$ with $h|_{M\times 0}:M\to M$ being the identity map.
\end{description}

%{\bf Convention:\/} in this paper $(\ss^{m+1}, g), m\geq 2$ i%s a globally hyperbolic spacetime with a smooth spacelike Cauchy surface $M^m$ and $h:M\times \R\to \ss$ is 
%the diffeomorphism as above.

Globally hyperbolic spacetimes form the most important class of spacetimes and one of the versions of the famous Strong Cosmic Censorship Conjecture of R.~Penrose says that all physically relevant spacetimes are globally hyperbolic~\cite{Penrose}. What happens inside of the black holes is not relevant for us living outside of them so the conjecture says that if you cut out the black holes along the event horizons then what is left is globally hyperbolic.

Let $\mathfrak N$ be the space of all future pointing null-geodesics in $(\ss, g)$ considered up to affine 
orientation preserving reparameterization.
A null geodesic $\gamma(t)$ intersects a Cauchy surface $M$ at a 
time $\ov t.$ We can decompose $T_{\gamma(\ov t)}\ss$ into a direct sum of $T_{\gamma(\ov t)}M$ and its 
$g$-orthogonal compliment. Since $M$ is spacelike and $\gamma$ is a null geodesic, 
the $T_{\gamma(\ov t)}M$-component of $\gamma'(\ov t)\neq {\bf 0}$ and it gives a point in the spherical 
tangent bundle $STM$ of $M.$ This identifies $\mathfrak N$ with $STM$. Since $M$ is spacelike we identify 
$TM$ with $T^*M$ and $\mathfrak N=STM=ST^*M.$ Low~\cite{LowLegendrian} showed that if $M'$ is another spacelike
Cauchy surface then $ST^*M'\to \mathfrak N\to ST^*M$ is a contactomorphism and thus his identification 
$\mathfrak N=ST^*M$  equips $\mathfrak N$ with the natural contact structure. 
Low's work~\cite{Low0, LowLegendrian} deals with $3+1$-dimensional $\ss$, but the fact and his proof hold in all 
dimensions, see also Natario and Tod~\cite[pages 252-253]{NatarioTod}.

The {\em sky} $S_p\mathfrak N$ of an event $p\in \ss$ is the Legendrian sphere of all light rays passing through $p\in \ss$. A sky $S_p$ is isotopic to the sphere fiber of $\mathfrak N=ST^*M\to M$ and two skies are called {\em unlinked \/}if they are isotopic to the link consisting of the pair of fibers over two distinct  points.

\subsection{Affine linking numbers and their application }
The classical linking number $\lk$ of two zero homologous submanifolds $P_1^{p_1}, P_2^{p_2}$ in $Q^{p_1+p_2+1}$ was introduced by Gauss and it is the signed number of intersection points of $P_1$ and a compact oriented submanifold bounded by $P_2.$ This does not depend on the order of the submanifolds $P_1, P_2.$ 

Clearly under the passage through a double point of a link under its homotopy $\lk$ changes by $\Delta_{\lk}$ which is the sign of the orientation frame given by concatenating the orientation frame of $P_1,$ of $ P_2$ and adding the vector going from the preimage of the double point on $P_1$ to the preimage on $P_2$ when you resolve the double point. 

Affine linking numbers $\alk$ were introduced by Rudyak and the first author~\cite{ChernovRudyakGT, ChernovRudyak} and they are defined as follows: for the pair of sphere fibers of $ST^*M$ over two distinct points the $\alk$ is zero and it increases by one under every positive passage through a transverse double point of the link. Rudyak and the first author~\cite{ChernovRudyak}[Theorem 2] showed that the affine linking number of a pair of skies in $ST^*M=\mathfrak N$ is a well defined $\Z$-valued invariant provided that $M$ is {\bf not} an odd-dimensional rational homology sphere. In the last case $\alk$ takes values in $\Z/\Im \deg$ where $\deg:\pi_m(M)\to \Z$ is the degree homomorphism. If two points are causally unrelated then $\alk$ of their skies equals zero. Moreover $\alk$ is symmetric.

For a pair of events $p_1, p_2$ in the future of $q$ connected by a generic timelike curve $\gamma=p(t)$ Rudyak and the first author proved that $\alk(S_{p_1}, S_q)-\alk(S_{p_2}, S_q)$ equals to the sum of the signs of the intersection points $\gamma\cdot\exp(N_q)$~\cite{ChernovRudyak}[Theorem 4], where $\exp(N_q)$ is the exponential of the null cone of $q.$ Every timelike curve can be made generic by a $C^{\infty}$-small perturbation so that it intersects $\exp(N_q)$ at the points where the exponential is immersed and there are no multiple points of $\exp(N_q)$ at such intersection points and moreover the intersection is transverse at all points.

Following Perlick~\cite{Perlick}[Section 2.8] we observe that for a timelike curve $\gamma$ that is a world line of a light source (with starting point $p_1$ and ending point $p_2$) and a point $q$ in the causal future of $\gamma$ the number of times an observer at $q$ sees light from $\gamma$ due to gravitational lensing is the number of light geodesics connecting a point of $\gamma$ with $q.$ Thus reversing the role of the points $p$ and $q$ in the above result~\cite{ChernovRudyak}[Theorem 4] we see that in the case where the affine linking number is a $\Z$-valued invariant we have that $\alk(S_{p_1}, S_q)-\alk(S_{p_2}, S_q)$ is the algebraic number of signs of times the point $q$ sees light from the light source with the world line $\gamma$ connecting $p_1$ and $p_2.$

Affine linking numbers can be used to estimate the number of times an observer sees the same star on the night sky, for example this happens in the famous Einstein Crosses of which about 50 examples were found so far, see Stern et al~\cite{Stern}.

\begin{thm}\label{theorem1}
Assume $(\ss, g)$ is globally hyperbolic and Cauchy surface of it is {\bf not} an odd dimensional rational homology sphere.
Let $\gamma:(-\epsilon, \epsilon)\to \ss$ be a small timelike curve that passes through $p$ at time moment $0.$ Assume that 
$\bigl (\alk( S_{\gamma(\epsilon)}, S_q) - \alk( S_{\gamma(-\epsilon),S_q})\bigr)
=N\in \Z$ for all small $\epsilon$ then the observer at $q$ sees light from $p$ coming from at least $|N|$ different directions.
\end{thm}

\begin{proof} The number of times an observer at $q$ sees light from $p$ is the limit as $\epsilon \to 0$ of the number of light geodesics connecting points of $\gamma$ with $q.$ The last quantity by the discussion before this Theorem is the difference of $\alk(S_{\gamma(\epsilon)}, S_q)$ and $\alk(S_{\gamma(-\epsilon), S_q})$.\end{proof}

As it was explained in~\cite{ChernovRudyak}[Theorem 4] if all the timelike sectional curvatures of $(\ss, g)$ are non-negative then the exponential of the null cone $N_q$ of $q$ is an immersion. Hence the terms in the computation of the intersection number of $\exp(N_q)$ and $\gamma$ do not cancel and hence the inputs into $\alk$ can not cancel as well. So we get the following Theorem.

\begin{thm}\label{theorem2}
Assume $(\ss, g)$ is globally hyperbolic and Cauchy surface of it is {\bf not} an odd dimensional rational homology sphere. Assume moreover that all the timelike sectional curvatures of $(\ss,g)$ are non-negative.
Let $\gamma:(-\epsilon, \epsilon)\to \ss$ be a small timelike curve that passes through $p$ at time moment $0.$ Assume that $\bigl ( \alk( S_{\gamma(\epsilon), S_q}) - \alk(S_{\gamma(-\epsilon),S_q})\bigr)
=N$ for all small $\epsilon$ then the observer at $q$ sees light from $p$ coming from exactly $N=|N|$ different directions.\qed
\end{thm}

\begin{rem}
Since the exponential of the null cone of $q$ is an immersion, all the intersection points of it with $\gamma$ are positive.
\end{rem}

\begin{rem}
An example of such a spacetime is any of the warped Lorentz spacetime products of a complete Riemann manifold of non-positive sectional curvature and $\R$ with any positive warping function.
\end{rem}

\centerline{\bf Acknowledgements}
This work was partially supported by a grant from the Simons Foundation
($\# 513272$ to Vladimir Chernov).

We are thankful to the anonymous referees for the careful reading and suggested improvements.

\end{document}